\documentclass[conference]{IEEEtran}
\usepackage{cite}
\usepackage{amsmath,amssymb,amsfonts,amsthm}
\usepackage{algorithmic}
\usepackage{graphicx}
\usepackage{textcomp}
\usepackage{xcolor}
\usepackage{graphicx}
\usepackage{booktabs}
\usepackage{nicefrac}
\usepackage{microtype} 
\graphicspath{{figures/}}
\usepackage[all]{nowidow}
\usepackage[utf8]{inputenc}
\usepackage{float}
\usepackage{hyperref}       
\usepackage{subfig}
\usepackage{multirow}
\usepackage{multicol}
\usepackage{xspace}
\usepackage[ruled,vlined]{algorithm2e}
\newcommand{\cegkr}{CEGkR\xspace}
\newtheorem{proposition}{Proposition}

\begin{document}

\sloppy
\title{Learning Distributed Controllers\\ for V-Formation}

\makeatletter
\newcommand{\linebreakand}{%
  \end{@IEEEauthorhalign}
  \hfill\mbox{}\par
  \mbox{}\hfill\begin{@IEEEauthorhalign}
}
\makeatother

\author{\IEEEauthorblockN{Shouvik Roy}
\IEEEauthorblockA{\textit{Department of Computer Science} \\
\textit{Stony Brook University} \\
Stony Brook NY, USA \\
shroy@cs.stonybrook.edu}
\and
\IEEEauthorblockN{Usama Mehmood}
\IEEEauthorblockA{\textit{Department of Computer Science} \\
\textit{Stony Brook University} \\
Stony Brook NY, USA \\
umehmood@cs.stonybrook.edu}
\and
\IEEEauthorblockN{Radu Grosu}
\IEEEauthorblockA{\textit{Cyber-Physical Systems Group} \\
\textit{Technische Universitat Wien} \\
Wien, Austria \\
radu.grosu@tuwien.ac.at}
\linebreakand
\IEEEauthorblockN{Scott A. Smolka}
\IEEEauthorblockA{\textit{Department of Computer Science} \\
\textit{Stony Brook University} \\
Stony Brook NY, USA \\
sas@cs.stonybrook.edu}
\and
\IEEEauthorblockN{Scott D. Stoller}
\IEEEauthorblockA{\textit{Department of Computer Science} \\
\textit{Stony Brook University} \\
Stony Brook NY, USA \\
stoller@cs.stonybrook.edu}
\and
\IEEEauthorblockN{Ashish Tiwari}
\IEEEauthorblockA{\textit{Microsoft Research} \\
San Francisco CA, USA \\
ashish.tiwari@microsoft.com}
}


\maketitle
\vspace{-2.5ex}
\begin{abstract}
We show how a high-performing, fully distributed and symmetric neural V-formation controller can be synthesized from a Centralized MPC (Model Predictive Control) controller using Deep Learning.  This result is significant as we also establish that under very reasonable conditions, it is impossible to achieve V-formation using a deterministic, distributed, and symmetric controller.  The learning process we use for the neural V-formation controller is significantly enhanced by CEGkR, a \emph{Counterexample-Guided $k$-fold Retraining} technique we introduce, which extends prior work in this direction in important ways. Our experimental results show that our neural V-formation controller generalizes to a significantly larger number of agents than for which it was trained (from~7 to~15), and   exhibits substantial speedup over the MPC-based controller.  We use a form of statistical model checking to compute confidence intervals for our neural V-formation controller's convergence rate and time to convergence.
\end{abstract}


\section{Introduction}
\label{sec:intro}

Designing distributed controllers is a challenging task, as the associated agents are typically attempting to achieve a global objective despite only having a local view of the global configuration.  They must therefore take actions based on incomplete information. Often it is not possible to optimize for global objectives using locally-optimal actions alone. High-performing distributed controllers may thus need to employ information-sharing among non-neighbors via complicated protocols, such as distributed consensus.

This state of affairs raises the following question.  Rather than manually designing distributed controllers, can we automatically learn them?  If so, how would we obtain the requisite training data without already having a solution for the distributed control problem in hand?

In this paper, we explore the use of a centralized controller, with global system knowledge, to generate the training data needed to learn a fully distributed neural controller. It is not obvious that this approach would work, since learning a high-performing distributed controller would require the learning process to (implicitly) figure out a way to gather information from non-neighbors.  Moreover, there is nothing in the training data that suggests how to perform such consensus tasks. A priori, we do not even know if such information sharing is possible in the distributed setting without explicit communication between agents.

To investigate this idea, we consider a particularly challenging multi-agent flight-formation problem: \emph{V-formation}, an emergent behavior of significant interest to the aerospace industry. The V-formation problem refers to the task of bringing a collection of agents from an arbitrary initial state to a state where they are all flying in a V-shape, with one agent leading the group and the others following on the left and right branches of the~V. V-formation provides numerous benefits. It is historically known for being energy-efficient due to the \emph{upwash benefit} an agent in the configuration enjoys from its frontal neighbor.  It also offers each agent a clear frontal view, unobstructed by any flock-mate.

The V-formation problem has been shown to be one of optimal control, which can be solved using model predictive control (MPC)~\cite{CAMACHO2007}. Section~\ref{sec:related} discusses various approaches that have been proposed to solve this problem. In particular, there exist centralized~\cite{ARES} and (partially) distributed~\cite{Lukina2019} solutions for achieving V-formation using MPC. None of these approaches, however, lead to a truly distributed solution for V-formation, i.e., without any form of consensus or information-sharing among non-neighbors. Specifically, the distributed solutions in prior work have three shortcomings. First, the distributed  controller~\cite{Lukina2019} uses a consensus round at the beginning of every time step, so that all agents agree on a consistent set of actions. 
This augmented controller performs tasks similar to leader election in the process.
Second, the controller uses \emph{adaptive neighborhood resizing} to enable agents to increase their neighborhood sizes to ensure convergence to a V-formation. MPC-based controllers can be computationally expensive, and increasing the neighborhood size increases the computational cost.  Third, each control step consisted of many ministeps where agents exchanged information and solved multiple optimization problems leading to a complicated procedure overall.

In this paper, we present \emph{Neural V-formation}, a new approach to the V-formation problem that uses Supervised Learning and a retraining technique we introduce called \emph{Counterexample-Guided $k$-fold Retraining} to learn a symmetric and fully distributed controller from a centralized, adaptive-horizon MPC controller~\cite{ARES}. By doing so, we achieve the best of both worlds: high performance on par with the MPC controllers, and high efficiency, which leads to real-time flight controllers.

Notably, we also show how our neural V-formation controller generalizes to a significantly larger number of agents (up to 15) than the number of agents on which it is trained (only 7). This generalization by our neural V-formation controller is achieved using only local neighborhood information and a local cost-function value, without any communication with other agents.  Our experiments demonstrate that attempting to use a distributed MPC controller (without explicit communication or consensus) to achieve this level of generalization does not yield satisfactory results and is computationally more expensive.

Figure~\ref{fig:architecture} provides an overview of our approach. A high-performing, centralized, adaptive-horizon MPC controller (CAMPC) provides the labeled training data to the learning agent: a symmetric and fully distributed neural controller (DNC) in the form of a Deep Neural Network (DNN). The training data consists of trajectories of state-action pairs, where a state contains the information known to an agent at a time step (its position and velocity, the positions and velocities of its neighbors, and the value of its local cost function), and the action (the label) is the acceleration assigned to that agent at that time step by the CAMPC controller. 

The key point here is that the CAMPC controller uses knowledge of the full state (positions and velocities of all agents) to find the optimal action for each agent, whereas the DNC controller is trained to compute the same output action only from information about its local state. The DNC has to do more than just a table lookup over the training data: it has to learn a function that uses only locally sensed data to compute the optimal action such that the same DNC works for all agents (and their local views) at all times.

The learning process we use for neural V-formation is significantly enhanced through the introduction of \emph{Counterexample-Guided $k$-fold Retraining} (CEGkR).  In this context, a counterexample is a trajectory along which the neural controller failed to achieve V-formation.  CEGkR utilizes the first $k$ states of such failed trajectories as retraining samples, repeating this process until the desired performance of the neural controller is attained.  In terms of verification of our neural controller, we use a form of statistical model checking \cite{Larsen14,Grosu14} to compute confidence intervals for its rate of convergence to a V-formation and for its time to convergence.

The rest of the paper is organized as follows. Section~\ref{sec:background} describes the model dynamics and the AMPC algorithm, including its cost function for V-formation. Section~\ref{sec:impossibility} presents impossibility results that illustrate the difficulty of achieving V-formation through distributed control. Section~\ref{sec:neural} introduces our distributed neural controller for V-formation, and the associated learning process, with a focus on Counterexample-Guided $k$-fold Retraining. Section~\ref{sec:exp} contains experimental results comparing our neural controller with MPC-based controllers, along with our statistical model checking results. Section~\ref{sec:related} discusses related work, while Section~\ref{sec:con} offers concluding remarks.

\begin{figure}[t]
    \centering
    \includegraphics[width=9cm]{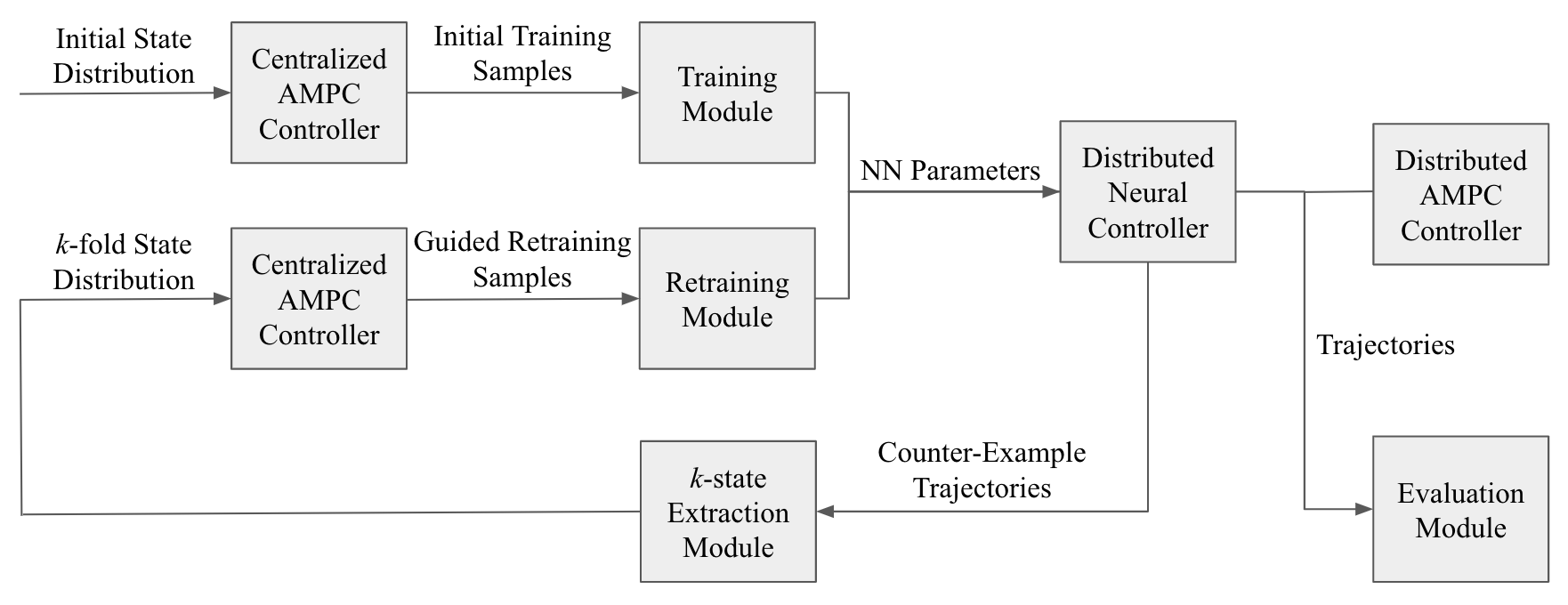}
    \caption{Neural V-formation Architecture}
    \label{fig:architecture}
    \vspace*{-2ex}
\end{figure}
\section{Background}
\label{sec:background}

We describe the model for our agents (including their equations of motion), the Centralized Adaptive-horizon MPC controller used to generate training data, and the distributed variant of the MPC controller with which we compare the performance of our neural V-formation controller.

The state of agent $i$ consists of four variables: a 2-dimensional vector $x_i$ giving the agent's position in 2D  space, and a 2-dimensional vector $v_i$ giving the agent's velocity. The state of a collection of $n$ agents is denoted $s = \{x_i,v_i \}_{i=1}^n$.  The control input, also called ``action'', for agent $i$ is a 2-dimensional acceleration denoted  $a_i$.

Let $x_i(t)$, $v_i(t)$ and $a_i(t)$ be the 2-dimensional positions, velocities and accelerations, respectively, of agent $i$ at time step $t$, $i \in \{1, \ldots, n\}$. The discrete-time equations of motion for agent $i$ are:
\begin{align}
    x_i(t+1) &= x_i(t) + dt \cdot v_i(t) \\
    v_i(t + 1) &= v_i(t) + dt \cdot a_i(t)
\end{align}
where $dt$ is the duration of a time step.

The goal of V-formation is to compute control actions (accelerations for the $n$ agents) that drive the system from an initial state (picked arbitrarily from some reasonable set of initial states) to a desired target state (a V-formation). We assume the desired final state is specified by a cost function, $J(s)$, that maps a state $s$ to a real-valued cost such that $J(s)=0$ exactly when $s$ represents the desired target state (V-formation), and $J(s) > 0$ otherwise. Further details about $J(s)$ are given below.

A \emph{Centralized Adaptive-horizon Model Predictive Control} (CAMPC) algorithm is proposed in~\cite{ARES}.  CAMPC generates action (acceleration) sequences using an adaptive prediction horizon $h$ to find the next action to execute towards the global optimum. CAMPC maintains multiple clones of the current state, and runs Particle Swarm Optimization (PSO)~\cite{kennedy95} on each of them. This allows it to call PSO for each clone with a different prediction horizon $h$.

CAMPC performs a system-wide minimization of the global cost function $J$ (defined in Eq.~\ref{eq:globalcost}) at each time-step to obtain an optimal action sequence of length $h$. The optimization is subject to the following constraints on the maximum velocities and accelerations:
\begin{align}
\label{eq:opt}
\forall i \in \{1,\,.\,.\,.\, n\},\,\, \Vert v_i(t)\Vert \leq v_{max} \,\land\, \Vert a_i (t) \Vert  \leq \rho \Vert v_i(t) \Vert
\end{align}
where $v_{max}$ is a constant and $\rho \in (0, 1)$.  PSO creates a swarm of $p$ particles uniformly at random within the given bounds on their positions and velocities. It then computes the fitness (the value of the cost function) of each particle. The fittest particle becomes a global best for the next iteration.  This procedure is repeated until the number of iterations reaches its maximum, a time limit is reached, or the cost function reaches its minimum value (i.e., a V-formation is achieved).

The adaptive prediction horizons are chosen such that the best-performing PSO instances succeed to decrease the objective cost by at least a pre-defined amount.  The adaptive-horizon feature allows PSO to escape from local minima by gradually increasing the MPC prediction horizon when necessary.  This provides convergence guarantees that would otherwise be impossible.

In~\cite{Lukina2019}, a distributed version of MPC is used to solve the V-formation problem, albeit with a reliance on a distributed consensus algorithm.  It deploys adaptive neighborhood resizing and an adaptive-horizon version of MPC to determine the optimal action (acceleration) for every agent at every time-step. Our comparative performance evaluation considers a modified version of this controller: \emph{Distributed Adaptive-Horizon Model Predictive Control} (DAMPC), which uses the adaptive-horizon feature of~\cite{Lukina2019}, but eschews any form of consensus.  This is to ensure a fair comparison with our neural controller, which is also ``consensus free''.  
DAMPC does not use the adaptive-neighborhood feature in~\cite{Lukina2019}; instead, it uses a fixed neighborhood size of 7 agents, just like our neural controller.  At any time-step, an agent's neighborhood consists of the 7 nearest agents, including itself.
Thus, in DAMPC, each agent $i$ computes the optimal action sequences for the agents in its neighborhood, and then uses the first acceleration in the sequence for itself. The accelerations are computed using PSO, as in CAMPC, except that the scope of the cost function is restricted to agents in $i$'s neighborhood, instead of all agents.

\begin{figure}[t]
\centering	
\subfloat[$\textit{CV}$]{\includegraphics[width=3cm]{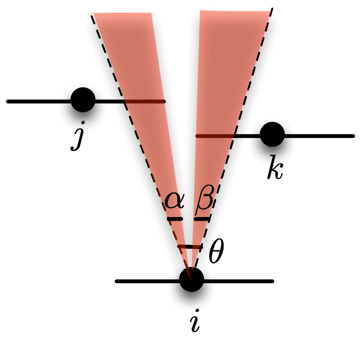}}
\subfloat[$\textit{VM}$]{\includegraphics[width=3cm]{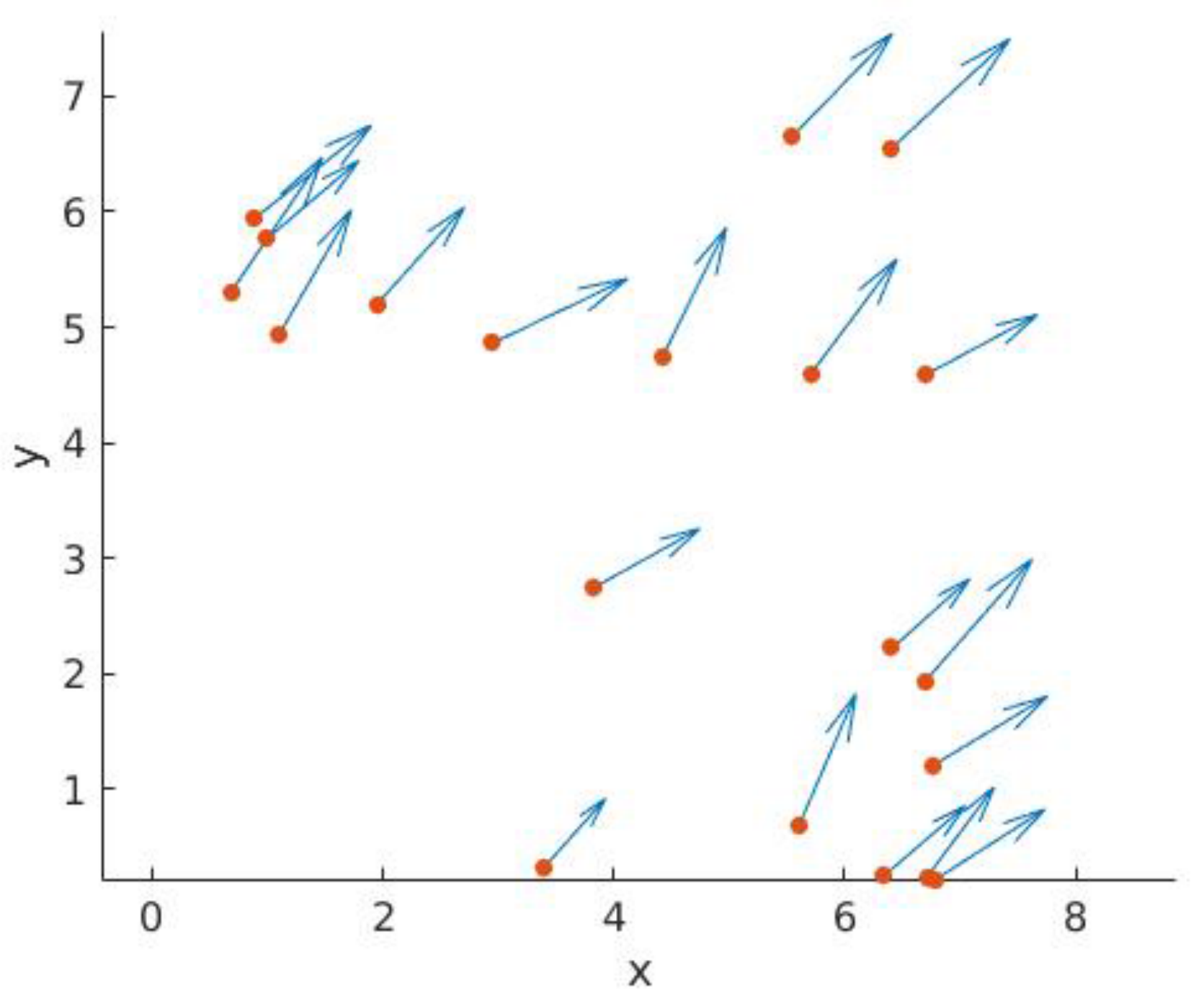}}
\subfloat[$\textit{UB}$]{\includegraphics[width=3cm]{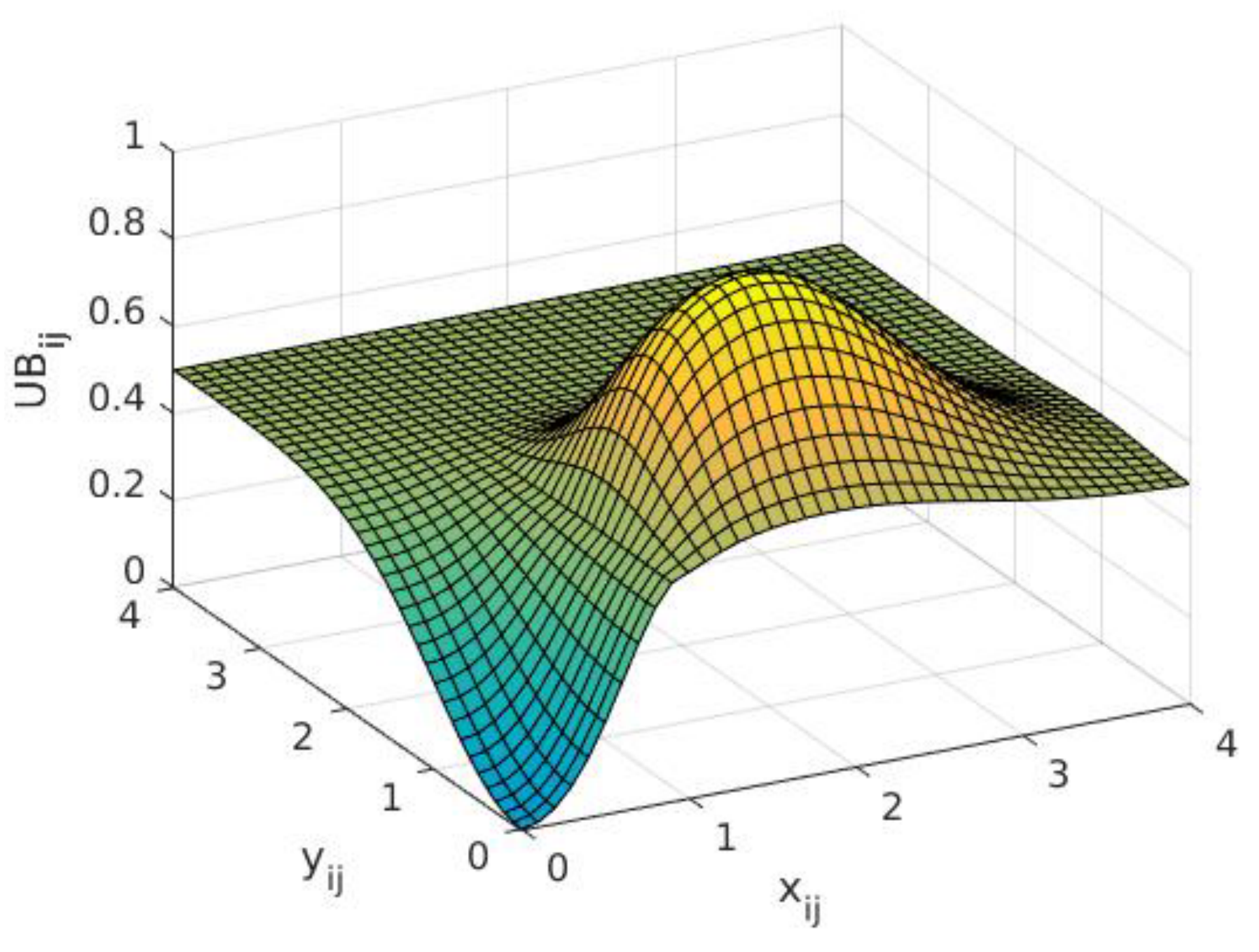}}

\vspace*{-1mm}
\caption{
(a) Agent $i$'s view is partially blocked by agents $j$ and $k$. Hence, its $\textit{CV} = \,(\alpha\,{+}\,\beta)/\theta$. 
(b) A flock and its unaligned agent velocities results in a velocity-matching metric $\textit{VM} = 6.2805$. In contrast, $\textit{VM} = 0$ when the velocities of all agents are aligned. 
(c) Illustration of the upwash benefit agent $i$ receives from agent $j$ depending on how it is positioned behind agent $j$.  Note that  agent $j$'s downwash region is directly behind it.
}
\label{fig:fitness}
\vspace*{-4mm}
\end{figure}

The global cost function $J(s)$, for state $s$, used in CAMPC for capturing V-formation, is defined in terms of the following metrics~\cite{YANG2016}.

\emph{Clear View}: An agent’s visual field is a cone with angle $\theta$ that can be blocked by the wings of other agents.  The clear-view metric $\textit{CV}(s)$ is defined as the sum over all agents $j$ of the percentage of agent $j$’s visual field that is blocked by other agents. 
Let $B_{ij}(x_i, v_i, x_j)$ be the part of the angle subtended by the wing of agent $j$ on the view of agent $i$ that intersects with agent $i$'s visual cone with angle $\theta$. Then, the clear view for agent $i$, $\textit{CV}_i(x, v)$, is defined as $|\cup_{j\neq i} B_{ij}(x_i, v_i, x_j)|/\theta$, and the total clear view, $\textit{CV}(s)$, is defined as $\sum_i \textit{CV}_i(x, v)$. The optimal value in a V-formation is $\textit{CV}^*{=}\,0$, as all agents have a clear view.

\emph{Velocity Matching}: $\textit{VM}(s)$ is defined as the accumulated differences between the velocity of each agent and all other agents, summed up over all agents. Formally. $\textit{VM}(s) = \sum_{i > j} ({\Vert v_i-v_j \Vert}/{(\Vert v_i \Vert + \Vert v_j \Vert)})^2$. The minimum value is \textit{VM}$^\ast = 0$ is attained when all agents have the same velocity. 

\emph{Upwash Benefit}: $\textit{UB}(s)$ is the sum of (the inverse of) each agent's upwash benefit. A trailing upwash effect is generated near the wingtips of an agent. An upwash measure is defined using a Gaussian model that peaks at the appropriate upwash regions. 
Let $h_{ij}$ be the projection of the vector $x_j - x_i$ along the wing-span of agent $i$. Similarly, let $g_{ij}$ be the projection of $x_j - x_i$ along the direction of $v_i$. Specifically, the upwash benefit $\textit{UB}_{ij}$ for agent $i$ coming from agent $j$ is given by
\begin{eqnarray}
\label{eq:ub}
    \textit{UB}_{ij} = \left\{
    \begin{array}{ll}
    \alpha  S(|h_{ij}|)\, \mathbb{G}_{ij} & \mbox{ if } |h_{ij}| \geq \frac{(4-\pi)w}{8} \wedge g_{ij} > 0 \\
    S(|h_{ij}|)\, \mathbb{G}_{ij} & \mbox{ if } |h_{ij}| < \frac{(4-\pi)w}{8} \wedge g_{ij} > 0 \\
    0 & \mbox{ otherwise}
    \end{array} \right.
\end{eqnarray}
where $S(z) = \mathtt{erf}(2\sqrt{2}(z - \frac{(4-\pi)w}{8}))$ is the error function, which is a smooth approximation of the sign function, $\mathbb{G}_{ij} = G(h_{ij}, g_{ij}, \mu_1, \Sigma_1)$. $G(y,z,\mu,\Sigma) = G_1( [\vert y \vert, \vert z \vert] - \mu, \Sigma)$ is a 2D-Gaussian shifted so that the mean is $\mu$, where $G_1(\vec{z}, \Sigma) = e^{(-\frac{1}{2}(\vec{z}^T \Sigma^{-1}\vec{z}))}$ is a 2D-Gaussian with mean at the origin. The parameter $w$ is the wing span, and $\mu_1 = [(12+\pi)w/16,1]$ is the relative position where upwash benefit is maximized. The total upwash benefit, $\textit{UB}_i$, for agent $i$ is $\sum_{j\neq 1}\textit{UB}_{ij}$. The maximum upwash an agent can obtain is upper-bounded by~1.  Since we are working with cost (that we want to minimize), we define $\textit{UB}(s) = \sum_i (1 - \min(\textit{UB}_i, 1))$. 
The upwash benefit in a V-formation is \textit{UB}$^\ast = 1$, as all agents, except for the leader, enjoy maximum upwash benefit.

The overall cost function $J(s)$ is be defined as a sum of squares:
\vspace{-1ex}
\begin{multline}
    \label{eq:globalcost}
    J(s) =  (\textit{CV}(s)-\textit{CV} ^\ast)^2+(\textit{VM}(s)-\textit{VM} ^\ast)^2 \\ +(\textit{UB}(s)-\textit{UB} ^\ast)^2
\end{multline}

For distributed controllers (DAMPC and the neural controller), we need to define a local cost function $J_i(s)$ for agent $i$.  It is the same as the global cost function, except that only agents in $i$'s neighborhood are considered.  This restriction applies to all aspects of the cost function.  For example, $CV_i(s)$ is the sum over agents $j$ in agent $i$'s neighborhood of the percentage of agent $j$’s visual field that is blocked by other agents in agent $i$'s neighborhood.

We consider an agent's neighborhood to consist of a fixed number $l$ of the nearest agents, including the agent itself.  Thus, agent $i$'s neighborhood consists of agent $i$ itself and the $l\,{-}\,1$ agents closest to it.  We take $l\,{=}\,7$ in our experiments. State $s$ is considered to be a V-formation if $J(s) \leq \varphi$ for a specified threshold $\varphi$.\footnote{The threshold $\varphi$ is a small positive constant chosen to allow for numerical errors due to floating-point computation, and also to allow for tiny perturbations that result in  formations which are visually indistinguishable from a V.}
\section{Impossibility Results}
\label{sec:impossibility}

Designing controllers that achieve V-formation when the controllers are distributed, symmetric, and deterministic is difficult.  This further motivates our proposed research on learning distributed and symmetric controllers for flight formation, including V-formation, from centralized controllers.

Distributed V-formation is an interesting and challenging problem. First note that a V-formation implicitly elects a leader. Hence a correct distributed algorithm will also solve the distributed leader election problem. It is known that there is no deterministic distributed leader election algorithm when all agents are identical \cite[chap.~3]{Attiya2004}.  This result, however, does not directly carry over to V-formation since the state of each agent consists of its (and its neighbors) spatial location, and two agents can never be identical (i.e., have the same spatial location for itself and its neighbors).  Nevertheless, most attempts to design deterministic distributed algorithms for V-formation will build in some form of spatial symmetry, and it is often possible to exploit this symmetry to devise initial configurations from which a proposed algorithm will fail to reach a V-formation.

First, V-formation inherits the issue in distributed systems that stems from the agents forming a disconnected partition.  For the next two results, we assume that: (A1) if the $k$ agents in the neighborhood of an agent (including itself) are in a perfect V-formation, then that agent would set its acceleration to $0$ (that is, it will maintain the formation).

\begin{proposition}
Under Assumption~(A1), if $N$ agents are spatially separated into two groups of $k$ and $N - k$ agents, where $k = \lfloor \frac{N}{2} \rfloor $, such that (1) each group is in a perfect V-formation, and (2) the $k-1$~nearest neighbors of any agent are in its own group, then a distributed procedure using neighborhood $k$ will fail to achieve a full V-formation on $N$ agents.
\end{proposition}

\begin{proof}
In the distributed case, the agents have no knowledge of the existence of the other group, and hence use an acceleration of $0$ to keep their current formation, according to Assumption~(A1).
\end{proof}

Agents partitioning into disconnected groups is not the only issue. Even when the neighborhood graph is connected, formations can look optimal locally, but remain unoptimal globally.


\begin{proposition}
Under Assumption~(A1), there exist initial configurations of $N$ agents such that starting from that configuration, a distributed procedure using neighborhood size $l=2$ will fail to achieve a full $V$-formation. 
\end{proposition}

\begin{proof}
Consider a perfect V-formation on three agents with one leader $A$ (at coordinate $(0,a)$), one agent $B$ on the left branch (at coordinate $(-b,0)$), and one agent $C$ on the right (at coordinate $(b,0)$), where $a, b$ are positive. Now, add the $4$-th agent $D$ at the position $(0,-a)$. Note that the position $(0, -a)$ experiences optimal upwash (coming from $B$ and $C$). Assume all agents have velocity $(0,1)$, and hence all agents are velocity matched. Agent $D$, however, does not have optimal clear view. If $l=2$, then every agent sees one other agent that is nearest to it.  Every such pair of agents, however, are in a local V-formation, so all agents set their acceleration to zero, by Assumption~(A1). Note that agent $D$ would not realize it doesn't have clear view unless it looks at at least two other agents, i.e., unless $l$ is at least $3$.
\end{proof}


These two propositions highlight two potential issues faced by a distributed approach: first, agents could get disconnected, and second, clear view is not a local property. However, what if the agents are connected and $l > 2$?
We present a scenario that demonstrates a third difficulty faced by a distributed procedure, namely the existence of multiple different optimal V-formations. We need some assumptions. We assume that: (A2)~if the velocity of an agent is aligned with the average velocity of all neighboring agents, then the controller picks an acceleration that is also aligned with that direction.  In other words, this assumption implies that if an agent is moving in the direction that is given by the average of the velocities of its neighbors, then it does not change its direction -- it can still speed up or slow down, but it keeps the direction of its motion unchanged.  This is a reasonable assumption for a controller since the controller is trying to achieve velocity matching and picking the average velocity is a commonly used strategy for this purpose.

We further assume that: (A3)~the controller is invariant to rotation of the coordinate axes; that is, just changing the frame of reference does not change the action computed by the controller for any configuration.  Assumption~(A3) is also a reasonable assumption. If a controller uses only the relative positions of its neighbors (with respect to its own position) and relative velocities of its neighbors (with respect to its own velocity), then such a controller can be seen to satisfy Assumption~(A3). 
If a controller satisfies Assumptions~(A2) and~(A3), then we can show that if every agent uses such a controller in a truly distributed manner to compute its own acceleration, then there are configurations that will never converge to a V-formation.

\begin{proposition}\label{prop3}
    If every agent's local controller satisfies Assumptions~(A2)--(A3), then there exists an initial configuration such that the trajectory of the multi-agent system starting from that initial configuration will never converge to a V-formation, even as the neighborhood graph remains connected and the agents use a neighborhood size greater than $2$. 
\end{proposition}
\begin{proof}
    Consider a total of $8$ agents placed on a circle equidistant from each other and moving radially outwards with equal speed. Let this be the initial configuration.  Let the neighborhood size be $3$. In this case, the neighborhood of each agent will include one neighbor on its left and one on its right.  Note the symmetry in this configuration.  The local configuration (involving $3$ agents) that is available to every agent is equivalent modulo rotation of the coordinate axes. Hence, by Assumption~(A3), if we know the action computed by any one agent, then we would know the action computed by all agents (by just rotating it appropriately).
      Therefore, let us focus on one agent.  Without loss of generality, assume that this agent, call it $A$, has position $(0,a)$ and velocity $(0,b)$, where $(0,0)$ is the center of the circle of radius $a$ on which all the $8$ agents lie.  Let $B$ and $C$ be the two neighbors of $A$. Therefore, $B$ has position 
      $(a/\sqrt{2},a/\sqrt{2})$ and $C$ has position
      $(-a/\sqrt{2},a/\sqrt{2})$. Furthermore,
      $B$ has velocity $(b/\sqrt{2},b/\sqrt{2})$ and $C$ has velocity
      $(-b/\sqrt{2},b/\sqrt{2})$.
      If we compute the average of the velocities of $A$, $B$, and $C$, we get the velocity
      $(0, (1 + \sqrt{2})b/3)$. 
      The direction of this average velocity is aligned with the velocity of $A$, and hence by Assumption~(A2), the acceleration for $A$  computed by the controller will be of the form $(0,c)$, for some $c$, which is acceleration in the radial direction. 
      By Assumption~(A3), the controller for every agent will pick an acceleration that is aligned with its current velocity.  Consequently, after one time step, the $8$ agents will continue to lie on a circle $A$ with center $(0,0)$, and with velocities that are pointing radially outwards or inwards.  We can now apply our argument again, and we can do so repeatedly to conclude that the $8$ agents will continue to lie on a circle forever. This shows that they will fail to converge to a V-formation.
\end{proof}

The issue highlighted in the above proof is that there are several optimal configurations, and different agents can decide to pick a different end configuration. In the above proof, agent $A$ concludes that it is moving in the ``correct'' direction and that its two neighbors $B$ and $C$ should change their direction to match its own direction.  And every agent, including $B$ and $C$, come to the same conclusion. This is because there are $8$ different V-formations -- one heading in each of the $8$ different directions.  And each agent picks a different final V-formation to target.  One might wonder if fixing the heading direction (or the target destination) would solve this issue: any such change in the problem definition surely invalidates Proposition~\ref{prop3}. However, we note that direction is not the only thing that can vary.  The speed with which each agent is moving in the final V-formation can also change. The different agents can not only pick different final speeds in their final $V$, but also different directions and even different leaders. This suggests that some coordination/consensus is required so that all agents agree and work toward the same final $V$-formation.
%
But building in any form of coordination and consensus is tedious and error-prone.


\section{Neural V-formation}
\label{sec:neural}

We learn a Distributed Neural Controller (DNC) for V-formation from trajectories obtained from a Centralized Adaptive-horizon Model Predictive Controller (CAMPC). Our learniong procedure makes use of a technique we call Counterexample-Guided $k$-fold Retraining (CEGkR), which uses counterexamples generated during testing of the neural controller as sources of new initial configurations for the CAMPC to generate additional training data. 

\subsection{Training a Distributed V-Formation Controller}
\label{sec:vform}
 
We use Deep Learning to synthesize a distributed and symmetric V-formation controller from the CAMPC controller (see Section~\ref{sec:background}), which generates the requisite training data in the form of trajectories leading to V-formation.  A  trajectory is a sequence of state-action pairs, where a state contains the information known to an agent at a time step (e.g., the positions and velocities of all agents in its neighborhood, including itself),
and the action (the label) is the acceleration assigned to that agent at that time-step by CAMPC.  We employ Supervised Learning to train our neural controller with the trajectories obtained from CAMPC.

The input features to the neural network are the 2-dimensional positions and velocities of all 7 agents in the agent's neighborhood and the value of the agent's local cost function.  Thus, the NN has 29 input features, and the input has the form $[p^x_0 \,\, p^y_0 \,\, v^x_0 \,\, v^y_0 \,\, p^x_1 \,\, p^y_1 \,\, v^x_1 \,\, v^y_1\cdot \cdot \cdot \cdot \, p^x_6 \,\, p^y_6 \,\, v^x_6 \, v^y_6 \, J]$, where $p^x_0$, $p^y_0$ and $v^x_0$, $v^y_0$ are the position and velocity coordinates, respectively, of the learning agent (i.e., the agent whose controller is being learned), $p^x_i$, $p^y_i$ and $v^x_i$, $v^y_i$ are the positions and velocities of the neighboring agents where $i = 1, \ldots, 6$, and $J$ is the local cost function of the learning agent.

We use CAMPC to generate trajectories, each with a duration of 50 time-steps. Let $\langle s_0,\vec{a}_0\rangle, \ldots, \langle s_{49},\vec{a}_{49}\rangle$ be a trajectory generated using CAMPC, where each $s_i$ is a $28$-dimensional state and $\vec{a_i}$ is the $14$-dimensional action computed for that state by CAMPC. For training the DNC controller, we obtain $50*7=350$ data points from each such trajectory, namely  $\langle T_j(s_i), J_j(s_i), \vec{a}_i[j]\rangle$, where $T_j(s_i)$ denotes agent $j$'s \emph{view} of state $s_i$, and $\vec{a}_i[j]$ is $j$'s 2-dimensional acceleration.

The view $T_j(s_i)$ is obtained by (a)~replacing absolute positions with positions relative to the position of agent $j$ (i.e., $p_k^x$ is replaced by $p_k^x - p_j^x$ and $p_k^y$ is replaced by $p_k^y - p_j^y$, for every agent~$k$); and 
(b)~permuting the indices of the agents so that the entries are in order of increasing distance from $j$.  Hence, agent $j$ is at index $0$, the nearest neighbor of $j$ in state $s_i$ is at index $1$, etc. Note that the relative position of agent $j$ is always $(0,0)$, so the first two entries of $T_j(s_i)$ are zero. Also, note that during all of the training (but not during the testing) performed in our experiments, the neighborhood size equals the total number of agents.  Thus, the local cost functions $J_j$ are equivalent to the global cost function $J$.

We learn a {\em single} neural V-formation controller from the state-action pairs of {\em all} agents. This yields a symmetric distributed controller, which we use for each agent during evaluation.  Note that the neural controller produces accelerations for only one agent, so it needs to be to run separately for each agent.

Our neural controller is a fully connected feed-forward deep neural network (DNN), with 5 hidden layers, 84 neurons per hidden layer, and with a sigmoid activation function.  To perform optimizations involving the MPC cost function, the Adam optimizer~\cite{adamopt} was used with the following settings: $lr\,{=}\,10^{-4}$, $\beta_1\,{=}\,0.9$, $\beta_2\,{=}\,0.999$, $\epsilon\,{=}\,10^{-8}$.  The number of trainable DNN parameters is 31,335, the batch size (number of samples processed before the model is updated) is 500, and the number of epochs (number of complete passes through the training dataset) used for training is 1000,.  The mean-squared error metric is used to measure training loss.  To train the neural networks, we use Keras~\cite{chollet2015keras}, which is a high-level neural network API written in Python and capable of running on top of TensorFlow. We used an iterative approach (based on the success rate of the neural controllers) for choosing the appropriate DNN hyperparameters and architecture.

\subsection{Counterexample-Guided $k$-fold Retraining}
\label{sec:cegkr}

We introduce a new counterexample-based retraining technique we call \cegkr to further improve the performance of the distributed neural controller (DNC) we obtain using the learning approach described in Section~\ref{sec:vform}. 
In the context of our V-formation investigation, \cegkr works as follows.  A retraining procedure first tests the neural controller by running it for $t=50$ time-steps, starting from $10^4$ randomly generated initial states.  We use a V-formation convergence threshold of $\varphi = 10^{-3}$ as the success criterion; i.e., the DNC successfully achieved V-formation if $J(s) \leq 10^{-3}$ at the end of the trajectory.  We refer to the failed trajectories as \emph{counterexamples}. Let $\tau = \langle s_0, \vec{a}_0\rangle, \ldots, \langle s_{49}, \vec{a}_{49}\rangle$ be a counterexample. Note that the neural controller fails to reach a V-formation (by the end of $\tau$) not only from the initial state $s_0$, but also from the subsequent states of $\tau$.

We do not know, however, whether states near the end of a failed trajectory are problematic for the neural controller, because there are not enough remaining time-steps in the trajectory to properly evaluate the controller's performance starting from those states.  Therefore, we pick a cutoff $k$ and use the first $k$ states in each counterexample as initial states to generate $k$ new training trajectories. We do this by running the CAMPC controller for 50 time-steps starting from each of these states.  Note that each counterexample leads to a total of $k\cdot 50\cdot 7$ new training data points for improving the neural controller.

After updating the controller using this new training data and using the same learning algorithm as in initial training, we test the updated controller by running it from a new batch of $10^4$ randomly generated initial states.  If the success rate (i.e., number of trajectories ending in a V-formation) for this batch is higher than in the previous round of testing, we perform another round of retraining.  Otherwise, the \cegkr retraining procedure is terminated.

As noted, \cegkr uses the first $k$ states from each counterexample trajectory to generate new training data.  Regarding the choice of $k$, we first observe that $k$ should not be too large, partly for the reason mentioned above, and partly because states near the end of the trajectory may be uninteresting, in the sense that they are encountered during testing only as the result of an accumulation of poor decisions made by the current controller earlier in the execution.  Our final improved controller will never encounter those states, so there is no benefit of using them for training.  For example, states where the flock has split into disconnected subflocks are uninteresting in this sense.

Secondly, $k$ should not be too small.  It is possible that the neural controller makes good decisions in the first several states, but later on, say in $s_{10}$, it chooses a poor action that leads to failure.  Using data for these later states during retraining provides the most benefit.  Note that these later states might not occur in a trajectory computed by CAMPC starting from an earlier state such as $s_0$ or $s_1$. In our experiments, we found that increasing $k$ (from 10 to 30 to 35) led to a concomitant increase in the success rate.  See Section~\ref{sec:cegkreval}.
\section{Experimental Results}
\label{sec:exp}

This section contains the results of our performance analysis of the distributed neural V-formation controller (DNC).  It specifically reports on the performance improvement due to \cegkr, compares the performance of DNC with DAMPC, and uses Statistical Model Checking to obtain confidence intervals for DNC's correctness/performance.

\subsection{\cegkr Performance Evaluation}
\label{sec:cegkreval}

Table~\ref{table:CEGER} demonstrates the DNC's performance improvement due to \cegkr. For the initial configurations used to generate the initial training samples, agent positions and velocities are uniformly sampled from $[0,5]^2$ and $[0.25, 0.75]^2$, respectively. The total number of initial training samples for each experiment is $init \cdot n$, where $init$ is the number of unique initial configurations, and $n$ is the number of agents.
An ``experiment'' is an instance of using the \cegkr methodology to generate a DNC.
A single training sample (trajectory) is comprised of $m=50$ discrete time-steps. For all experiments, $n=7$, and for Experiment~1, $init=23,000$, and for Experiments~2, 3 and~4, $init=50,000$. As we show below, increasing $init$, which thereby increase the total number of training samples, increases the success rate; i.e., the rate of reaching V-formation.


For Experiment~1, Run~1, we initially train our DNC with 161,000 training samples, and perform no retraining. This version of DNC achieves a success rate of 85.07\% on $10^4$ test cases, which are generated from the same distribution used for initial training. For Run~2, we take the first 10 states of all Run~1 failed test cases and use them as initial states for CAMPC to use to generate new trajectory data, i.e.,  the guided training samples. The total number of guided retraining samples is $\textit{f} \; n \; k$, where $f$ is the number of failed test cases. For example, in Run~1, $f = 1493$ and $k = 10$, so the number of guided retraining samples is $1,493 \cdot 7 \cdot 10 = 104,510$. Retraining with these samples leads to a 4\% increase in the success rate. As described in Section~\ref{sec:cegkr}, we repeat this procedure until there is no improvement in the DNC success rate.

\begin{table}[t!]
\caption{CEGkR results for V-formation based on $10^4$ test cases}
  \label{table:CEGER}
  \centering
  \resizebox{\linewidth}{!}{%
  \begin{tabular}{ccccc}
    \toprule
     Retraining  & \# Guided     & Success     & Median\\
     Run Id       & Retraining Samples & Rate ($\%$) & Final $J$\\
    \hline
    & Experiment 1 :& 161,000 Initial Training Samples & \\
    \hline
    Run 1   &      0     & 85.07  & 0.0002747 \\
    Run 2   & 104,510  & 89.03  & 0.0000499 \\
    Run 3   & 76,790  & 91.12  & 0.0000299 \\
    Run 4   & 62,160  & 91.12  & 0.0000299 \\
    \hline
    & Experiment 2 :& 350,000 Initial Training Samples & \\
    \hline
    Run 1   &      0     & 90.08  & 0.0000351 \\
    Run 2   & 69,440  & 91.10  & 0.0000300 \\
    Run 3   & 62,300  & 92.22  & 0.0000231 \\
    Run 4   & 54,600  & 92.91  & 0.0000222 \\
    Run 5   & 49,630  & 93.04  & 0.0000222 \\
    Run 6   & 48,720  & 93.04  & 0.0000222 \\
    \hline
    & Experiment 3 :& 350,000 Initial Training Samples & \\
    \hline
    Run 1   &      0     & 90.08  & 0.0000351 \\
    Run 2   & 208,320 & 92.16  & 0.0000249 \\
    Run 3   & 164,640  & 93.33  & 0.0000221 \\
    Run 4   & 140,070  & 94.01  & 0.0000200 \\
    Run 5   & 125,790  & 94.95  & 0.0000155 \\
    Run 6   & 106,050  & 94.95  & 0.0000155 \\
    \hline
    & Experiment 4 :& 350,000 Initial Training Samples & \\
    \hline
    Run 1   &      0     & 90.08  & 0.0000351 \\
    Run 2   & 243,040 & 93.02  & 0.0000235 \\
    Run 3   & 171,010  & 93.88  & 0.0000213 \\
    Run 4   & 149,940  & 94.51  & 0.0000202 \\
    Run 5   & 134,505  & 95.16  & 0.0000153 \\
    Run 6   & 118,335  & 95.16  & 0.0000151 \\
    \hline
  \end{tabular}}
  \vspace*{-2ex}
\end{table}

Experiment~2 is similar to Experiment~1.  The only difference is that Experiment~2 has approximately twice the number of initial training samples as compared to Experiment~1, which gives it an improved initial success rate. Experiments~3 and~4 use the same set of initial training samples as Experiment~2; the difference is that they use $k=30$ and $k=35$, respectively, instead of $k=10$.

Table~\ref{table:CEGER} demonstrates the benefit of CEGkR, which include the following.
(1)~CEGkR always improves the performance of the learned controller.  Specifically, Run~$2$ in every experiment shows significant improvement over Run~$1$. (2)~As expected, CEGkR does not improve the success rate forever; rather the success rate eventually plateaus.
(3)~Using higher values for $k$ in the CEGkR retraining loop improves the quality of the learned controller: the success rate in Experiment~4 is better than that in Experiment~3, which is better than that in Experiment~2.  Note that increasing $k$ also increases the size of the training data and therefore the cost of retraining.

Table~\ref{table:cegkrcomparison} presents a comparative evaluation of the performance of neural controllers obtained with and without CEGkR. The same number of training samples are used to train both controllers. We use the DNC obtained from Experiment~4 in Table~\ref{table:CEGER} as our neural controller with CEGkR.  For the non-CEGkR controller, we trained it using 1,116,830 training samples, which is equal to the total number of training samples (initial training samples + guided retraining samples) used for training the DNC with CEGkR.  The results show that using \cegkr offers a clear advantage, as the \cegkr controller has a consistently higher success rate as the number of agents generalizes beyond~7.

\begin{table}[t]
\caption{Performance comparison for DNC with and without CEGkR on $10^4$ runs}
  \label{table:cegkrcomparison}
  \centering
  \resizebox{0.65\linewidth}{!}{%
  \begin{tabular}{ccccccc}
    \toprule
    Number of& \cegkr & Non-\cegkr \\
    Agents & Success Rate($\%$) & Success Rate($\%$) \\
    \hline
    7  & 95.16 & 90.51\\
    8  & 94.57 & 89.03\\
    9  & 93.78 & 87.66\\
    10 & 93.05 & 85.25\\
    11 & 93.05 & 84.10\\
    12 & 92.67 & 81.98\\
    13 & 91.38 & 80.24\\
    14 & 87.35 & 78.47\\
    15 & 84.25 & 74.72\\
    16 & 73.40 & 65.39\\
    \hline
  \end{tabular}}
\end{table}

\subsection{Comparing the Performance of DNC vs DAMPC}

The experiments in this section compare the performance of DNC with the distributed adaptive-horizon MPC controller (DAMPC). We focus on DAMPC and DNC because unlike CAMPC, they both rely on sensing only and not communication.
The DNC we use from here onwards is the one obtained from Experiment~4 in Table~\ref{table:CEGER}. For determining DNC's success rate, we modify the convergence threshold and number of time steps that the controller runs to be proportional to the number of agents $n$. Specifically, we use a convergence threshold of $ \left( n/7 \right) \varphi$ and a number of time-steps of $\left( n/7 \right) m$, where $\varphi=10^{-3}$ and $m=50$. We observed experimentally (visually) that this proportional increase in the threshold is justified.  The rationale for increasing the number of time-steps is that with an increasing $n$, the DNC will take longer to converge.

The DAMPC controller is presented in Section~\ref{sec:background}; recall that it is a variant of the one presented in~\cite{Lukina2019}.  The adaptive-horizon feature is used with the prediction horizon $h$ restricted to the interval $[1, 3]$.

\begin{table}[t!]
\caption{Performance Comparison: DNC vs DAMPC on $10^4$ runs}
  \label{table:gen}
  \centering
  \resizebox{0.75\linewidth}{!}{%
  \begin{tabular}{ccccc}
    \toprule
    & \multicolumn{2}{c}{DAMPC} & \multicolumn{2}{c}{DNC} \\
    \hline
    Number of & Success & Avg. Conv. & Success & Avg. Conv. \\
    Agents & Rate($\%$) & Time & Rate($\%$)& Time \\
    \hline
    7  & 89.84 & 20.11 & 95.16 & 19.69 \\
    8  & 85.16 & 21.73 & 94.57 & 20.05 \\
    9  & 79.04 & 24.27 & 93.78 & 20.58 \\
    10 & 75.37 & 24.52 & 93.05 & 22.16 \\
    11 & 70.91 & 26.03 & 92.67 & 23.89 \\
    12 & 66.82 & 27.86 & 91.38 & 25.23 \\
    13 & 61.58 & 32.23 & 89.97 & 27.77 \\
    14 & 52.49 & 34.87 & 87.35 & 29.24 \\
    15 & 41.75 & 39.71 & 84.25 & 34.31 \\
    16 & 34.03 & 39.84 & 73.40 & 39.05 \\
    \hline
  \end{tabular}}
  \vspace*{-2ex}
\end{table}

Table~\ref{table:gen} demonstrates the generalization capabilities of DNC (from~7 to~16 agents), and compares its performance with that of DAMPC. While increasing the number of agents from $7$ to $16$, the neighborhood size is fixed at~7. The main observations from Table~\ref{table:gen} are the following. (1)~DNC consistently outperforms DAMPC, thus demonstrating that our approach for learning distributed controllers from training data generated by a centralized controller produces a very effective distributed controller, one that outperforms a distributed controller designed following the well-established MPC-based approach.
(2)~DNC's average convergence time is considerably smaller than that for DAMPC.  Note that the \emph{convergence time} is the time when the global cost function first drops below the success threshold $\varphi$.  Since the calculation of average convergence time only uses successful runs (ignoring the failed runs), it follows that not only does DNC achieve success more often, it does so in fewer steps.  This means it is better than DAMPC at avoiding wrong decisions that lead to local minima.

An important advantage of the neural controller over CAMPC and DAMPC is that the former is much faster at generating the action at every time-step for each agent.  Executing a DNC requires a modest number of arithmetic operations, whereas executing an MPC controller requires simulation of a model and a controller over the prediction horizon.

In our experiments, on average, CAMPC and DAMPC take 1,730~msec and 524~msec of CPU time, respectively, whereas the DNC only takes 1.5~msec.  These results are averages over $10^4$ runs with 7 agents.  Although multiple instances of DNC are needed (one per agent), they all run in parallel, so it is reasonable to compare the CPU time of CAMPC with that for one instance of DNC.  Even if we consider the total CPU time for all instances of DNC, it is much less than CAMPC.

\begin{table}[t]
\caption{Robustness Performance for DNC on $10^4$ runs}
  \label{table:robustness}
  \centering
  \resizebox{\linewidth}{!}{%
  \begin{tabular}{lccc}
    \toprule
     Config.\ Space &  \# Agents  & Success  Rate ($\%$) & Avg. Convergence Time \\
    \hline
    Pos: $[0,6]^2$      & 7  & 94.28 & 19.88 \\
    Vel: $[0.4, 0.8]^2$ & 15 & 82.19 & 39.12 \\
    \hline
    Pos: $[0,8]^2$        & 7  & 91.84 & 20.54 \\
    Vel: $[0.35, 0.95]^2$ & 15 & 78.33 & 40.01 \\
    \hline
    Pos: $[0,10]^2$     & 7  & 87.63 & 20.93 \\
    Vel: $[0.1, 0.9]^2$ & 15 & 75.46 & 39.43 \\
    \hline
  \end{tabular}}
  \vspace*{-2ex}
\end{table}

\subsection{Evaluating Robustness of Distributed Neural Controller}
\label{sec:robustness}

We also demonstrate that our DNC is robust to variations in the initial conditions: it performs well even from initial states well beyond the range of initial states on which it was trained.  Recall that during training, the positions and velocities are uniformly sampled from $[0,5]^2$ and $[0.25, 0.75]^2$, respectively.  We test the controller on initial states selected from three other configuration spaces (i.e., ranges of initial states), which are defined in Table~\ref{table:robustness}.  The initial positions and velocities are uniformly sampled from these ranges.  The table also shows the number of agents, the percentage of successful executions, and the average convergence time.  The results in each row are averages over $10^4$ runs.

\begin{figure}[t]
\centering
\subfloat{\includegraphics[width=7cm]{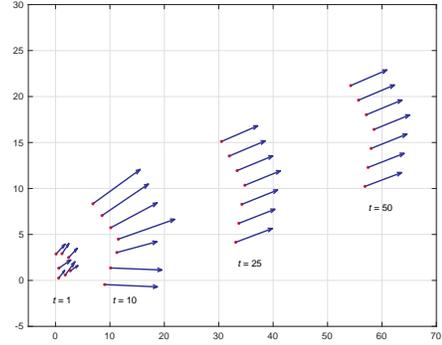}}
\caption{Snapshots of V-formation with 7 agents using DNC}\label{fig:snapshot}
\vspace*{-2ex}
\end{figure}

When we move from the configuration space used during initial training to the third configuration space, the size of the set of possible initial positions expands by a factor of $(10/5)^2$, and the size of the set of possible initial velocities expands by a factor of $(0.8/0.5)^2$; hence there is an overall expansion factor of $\sim\!10$ in the initial state space.  This means that the probability that an initial state picked randomly from the third configuration space also lies inside the initial training configuration space is approximately $0.1$.  Thus, among the $10^4$ runs, there are only around $10^3$ runs on which we definitely expect a high rate of success.  

The actual success rate is much better than this argument suggests.  The success rate decreases from $95\%$ (in Table~\ref{table:gen}) to $87\%$ for 7~agents, and from $84\%$ to $75\%$ for 15~agents.  This is roughly a $10$\% decrease, much less than the $90$\% drop that would occur if the NN controller did not generalize its training in order to perform well from initial states beyond those used during training.

Figure~\ref{fig:snapshot} shows the progression of seven agents starting from initial positions randomly selected from the range $[0,5]$ till until they successfully converge to a V-formation. At $t=25$, we can observe that the agents have reached a V-formation, and thus the convergence time is 25.

\subsection{Statistical Model Checking Results}

We use Monte Carlo (MC) approximation as a form of Statistical Model Checking~\cite{Larsen14,Grosu14} to compute confidence intervals for the DNC's success rate for convergence to V-formation and for the (normalized) convergence time. The main idea of MC is to use $N$ random variables, $Z_1, \ldots, Z_N$, also called samples, IID distributed according to a random variable $Z$ with mean $\mu_Z$, and to take the sum $\tilde \mu_Z = (Z_1 + \ldots + Z_N)/N$ as the value approximating the mean $\mu_Z$.  Since an exact computation of $\mu_Z$ is almost always intractable, an MC approach is used to compute an ($\epsilon, \delta$)-approximation of this quantity.

\emph{Additive Approximation}~\cite{Thomas2004} is an ($\epsilon, \delta$)-approximation scheme where the mean $\mu_Z$ of an RV $Z$ is approximated with absolute error $\epsilon$ and probability $1-\delta$:
\begin{equation}
\label{eq:epsilon-delta}
    Pr[\mu_Z - \epsilon \leq \tilde{\mu}_Z \leq \mu_Z + \epsilon] \geq 1 - \delta
\end{equation}
where $\tilde{\mu}_Z$ is an approximation of $\mu_Z$. An important issue is to determine the number of samples $N$ needed to ensure that $\tilde{\mu}_Z$ is an ($\epsilon, \delta$)-approximation of $\mu_Z$.  If $Z$ is a Bernoulli variable expected to be large, one can use the Chernoff-Hoeffding instantiation of the 
Bernstein inequality and take $N$ to be ${N = 4 \ln(2/\delta)/\epsilon^2}$, as in~\cite{Thomas2004}. 
This results in the following $\emph{additive approximation algorithm}$ \cite{Grosu14}:

\begin{algorithm}[h]
\KwIn{($\epsilon$, $\delta$) with $0 < \epsilon < 1$ and $0 < \delta < 1$}
\KwIn{Random variables $Z_i$, IID}
\KwOut{$\tilde{\mu}_Z$ approximation of $\mu_Z$}

$N = 4 \ln(2/\delta)/\epsilon^2$;\\
\For{(i=0; i $\leq N$; i++)} {
$S = S + Z_i$;\\
}
$\tilde{\mu_Z} = S/N$;
\Return{$\tilde{\mu}_Z$}\;
    \caption{{\bf Additive Approximation Algorithm} \label{algo:AAA}}
\end{algorithm}
We use this algorithm to obtain a joint ($\epsilon, \delta$)-approximation of the mean success rate and mean normalized convergence time for the DNC trained using CEGkR.  Each sample $Z_i$ is based on the result of an execution obtained by simulating the system starting from a random initial state, and we take $Z=(B,R)$, where $B$ is a Boolean variable indicating whether the agents converged to a V-formation during the execution, and $R$ is a real value denoting the normalized convergence time in the execution. The normalized convergence time is the time when the global cost function first drops below the success threshold and remains below it for the rest of the execution, measured as a fraction of the total duration of the simulation.  The assumptions about $Z$ required for validity of the additive approximation hold, because RV $B$ is a Bernoulli variable, the success rate is expected to be large (i.e., closer to 1 than to 0), and the proportionality constraint of the Bernstein inequality is also satisfied for RV $R$.

In these experiments, the initial states are sampled from the same uniform random distributions as in Section \ref{sec:cegkreval}, and we set $\epsilon = 0.01$ and $\delta = 0.0001$, to obtain $N =$ 396,140.  We perform the required set of $N$ simulations for different numbers of agents, ranging from 7 to 16.

Table~\ref{table:smc} presents the results, specifically, the ($\epsilon, \delta$)-approximations $\tilde \mu_{SR}$ and $\tilde \mu_{CT}$ of the mean success rate and mean normalized convergence time, respectively.  While the results for the success rate are (as expected) numerically similar to the results in Table \ref{table:cegkrcomparison}, the results in Table~\ref{table:smc} are much stronger, because they come with the guarantee that they are ($\epsilon, \delta$)-approximations of the actual mean values.


\addtolength{\tabcolsep}{6pt} 
\begin{table}[t]
\caption{SMC results for DNC success rate/convergence time; $\epsilon = 0.01$, $\delta = 0.0001$}
  \label{table:smc}
  \centering
  \resizebox{0.6\linewidth}{!}{%
  \begin{tabular}{ccc}
    \hline
     \# Agents & $\tilde \mu_{SR}$ & $\tilde \mu_{CT}$\\
    \hline
    7  & 0.9511 & 0.3942 \\
    8  & 0.9453 & 0.4024 \\
    9  & 0.9382 & 0.4128 \\
    10 & 0.9305 & 0.4426 \\
    11 & 0.9262 & 0.4770 \\
    12 & 0.9141 & 0.5058 \\
    13 & 0.8994 & 0.5560 \\
    14 & 0.8727 & 0.5852 \\
    15 & 0.8419 & 0.6874 \\
    16 & 0.7338 & 0.7822 \\
    \hline
  \end{tabular}}
  \vspace*{-2ex}
\end{table}

\section{Related Work}
\label{sec:related}

Distributed control/coordination has been used extensively in multi-agent systems. Distributed controllers are typically designed by hand for specific objectives, and they are often very clever about what information is exchanged between agents and how that information is used to update local state~\cite{fb-jc-sm:08c,jwd-af-fb:10v}. Informally, coordination is required when the cost function is non-separable. A cost function is \emph{separable} if it does not contain terms that couple the states of two different neighbors~\cite{dunbar2006,dunbar2004,Balas2006}. Here we take the novel view of \emph{learning} distributed controllers from data generated using a centralized controller, while avoiding coordination. We apply it to a problem whose cost is clearly not separable and involves tight coupling of state vectors of all pairs of agents.

Previous work on V-formation, including approaches based on centralized and distributed model-predictive control, have been considered in~\cite{YANG2016,ARES,Lukina2019}.
%
Other related work, including~\cite{D'Andrea2003,Fowler2002,Ye2017}, focuses on distributed controllers for flight formation (of moving-wing aircraft) that operate in an environment where the multi-agent system is already in the desired formation and the distributed controller's objective is to maintain formation in the presence of disturbances. A distinguishing feature of these approaches is the particular formation they are seeking to maintain, including half-V~\cite{Fowler2002}, ring and torus~\cite{D'Andrea2003}, and a leader-follower formation~\cite{Ye2017}. In~\cite{Moreno2017}, MPC-inspired approaches to system self-adaptation are considered, including the Proactive Latency-aware Approach (PLA)~\cite{Moreno2015}. The PLA problem is designed as a Markov decision process, where a sequence of actions is computed from the current state for the length of the prediction horizon. 

In terms of related work on counter-example-guided retraining, Dreossi et al.~\cite{Dreossi2018} propose an approach called \emph{counter-example guided data augmentation} to improve the performance of machine learning models. They use synthetically generated data items that are misclassified by the ML model to augment the training data sets. In \cite{Carr2019}, the authors use counter-example guided retraining as part of their strategy for synthesizing partially observable Markov decision processes (POMDPs). Claviere et al.~\cite{Claviere2019} use counter-example guided training for trajectory-tracking control of robotic vehicles.

The CEGkR retraining approach shares the same high-level philosophy underlying these approaches, but there are subtle differences in the way counter-examples are generated. In~\cite{Claviere2019}, counter-examples are generated using falsification of desired temporal properties about a closed-loop system, whereas in our approach, safety constraints, if any, are included in the cost function and {\em{multiple}} retraining data points are generated from {\em{one}} counterexample. Further, our goal for retraining is to learn a distributed controller, rather than an NN representation of an existing controller.

In terms of deep-learning methodologies for synthesizing distributed controllers, deep reinforcement learning is used in~\cite{CondeLT17} for designing controllers for UAVs that reach time-varying formations. They also use a DNN to estimate how good a state is, so the agent can choose actions accordingly. Deep reinforcement learning is also used in~\cite{Yang2018} to generate a controller for UAVs in uncertain environments. As the multi-agent learning efficiency is constrained by the high-dimensional and continuous action spaces, a methodology is presented in~\cite{Yang2018} to slice the action spaces into a number of tractable fractions to achieve efficient convergences of optimal policies in continuous domains. Graph neural networks are deployed in~\cite{tolstaya2019} to learn a distributed controller for a drone swarm capable of achieving flocking formation. The learned controller, which is synthesized by imitating the policy of a centralized controller, exploits information from distant teammates using only local communication interchanges.
\section{Conclusion}
\label{sec:con}

We have presented a new learning-based approach for designing distributed controllers that uses centralized controllers to generate the training data, in a teacher-learner fashion.  The data generated by a centralized controller undergoes a transformation to yield the requisite training data, a transformation defined by the information available to an agent in the distributed setting. During training, we use counterexample-guided $k$-fold retraining to generate additional data points to train the distributed controller. We demonstrated the power of this approach by developing a distributed neural controller for the V-formation problem, and used Statistical Model Checking to reason about the controller's correctness.

The V-formation problem is particularly challenging. We showed that a symmetric deterministic distributed controller does not exist under certain reasonable assumptions. 
This motivates the use of a data-driven approach to automatically synthesize such controllers. The general idea of learning distributed controllers from training data generated by centralized controllers is promising. We believe that our approach will generalize to any distributed control synthesis problem whose objective is specified by a state-based cost function.  Investigating its performance on other applications, and exploring enhancements to our learning-based approach to distributed controller design are directions for future work. 

\bibliographystyle{ieeetr}
\bibliography{biblio}
\end{document}